\documentclass[review,onefignum,onetabnum]{siamonline190516-copy}
\usepackage{amsmath,amssymb,amsfonts,euscript,graphicx}

\theoremstyle{plain}
\newtheorem{remark}[theorem]{Remark}

\def \R{{\mathbb R}}
\def \N{{\mathbb N}}
\def \A{{\mathcal A}}
\def \B{{\mathcal B}}

\def \oc{{\operatorname{c}}}

\newcommand{\nihil}[1]{}

\renewcommand{\theenumi}{(\roman{enumi})}

\begin{document}

\title{Minimizing the Repayment Cost of Federal Student Loans}

\author{Paolo Guasoni\thanks{Dublin City University, School of Mathematical Sciences, Glasnevin, Dublin 9, Ireland, 
Dipartimento di Statistica, Universit\`a di Bologna, Via Belle Arti 41, 40126 Bologna, Italy, 
email: \texttt{paolo.guasoni@dcu.ie}. Partially supported by SFI (16/SPP/3347 and 16/IA/4443).}
\and
Yu-Jui Huang\thanks{
University of Colorado, Department of Applied Mathematics, Boulder, CO 80309-0526, USA, email: \texttt{yujui.huang@colorado.edu}. Partially supported by NSF (DMS-2109002).}
}
\maketitle

\begin{abstract}
Federal student loans are fixed-rate debt contracts with three main special features: (i) borrowers can use income-driven schemes to make payments proportional to their income above subsistence, (ii) after several years of good standing, the remaining balance is forgiven but taxed as ordinary income, and (iii) accrued interest is simple, i.e., not capitalized. 
For a very small loan, the cost-minimizing repayment strategy dictates maximum payments until full repayment, forgoing both income-driven schemes and forgiveness. For a very large loan, the minimal payments allowed by income-driven schemes are optimal. For intermediate balances, the optimal repayment strategy may entail an initial period of minimum payments to exploit the non-capitalization of accrued interest, but when the principal is being reimbursed maximal payments always precede minimum payments.
Income-driven schemes and simple accrued interest mostly benefit borrowers with very large balances.
\end{abstract}



\thispagestyle{empty}



\section{Introduction}

Student loans today are a leading component of non-mortgage household debt, accounting for liabilities of over \$1.7 trillion, triple their size fifteen years ago, and exceeding both auto loans (\$1.3 trillion) and credit card balances (\$0.9 trillion).\footnote{See 
\url{https://fred.stlouisfed.org/series/BOGZ1FL153166220Q},
\url{https://fred.stlouisfed.org/series/BOGZ1FL153166400Q},
\url{https://fred.stlouisfed.org/series/BOGZ1FL153166110Q}} 
The bulk of such loans is issued by the federal government through various schemes, which give borrowers a plethora of repayment options, such as income-driven repayment, consolidation, deferral, forbearance, forgiveness -- even the recent pandemic suspension. 
The combination of all these features makes student loans very unusual debt contracts, more akin to financial derivatives than fixed-income products, and leaves individual borrowers with contracts that are extremely hard to understand and manage optimally, even for graduates with advanced degrees.

Federal loans make funds available to students to cover their tuition and living expenses. A few months after graduation or un-enrollment, students are responsible for repaying their debts, and here the repayment puzzle begins. 
Student loan balances grow at a national fixed rate (which depends on the origination date of the loan) and can be repaid in full over a fixed horizon (typically ten years), while a borrower is free to make additional payments at no penalty, similar to a mortgage.
Yet, student loans have three main peculiar features that distinguish them from other debt.

First, borrowers can enroll in income-driven repayment schemes, whereby monthly payments are due only if their income is above a certain subsistence threshold and are proportional to the amount by which it exceeds such threshold.
Second, after a number of years of qualifying payments (usually 20 or 25), the remaining balance is forgiven, but the forgiven amount counts as taxable income, hence generates a large final tax liability.\footnote{
The American Recovery Act passed in March 2021 made student-loan forgiveness tax-free through 2025, but not from 2026 onwards. An exception is the Public Sector Loan Forgivness Program, for which forgiveness is tax-free after ten years, without expiration.}
Third, if income-driven repayments are not enough to cover interest on the loan, accrued interest is not added to the principal balance, i.e., interest is not capitalized.\footnote{Capitalization of interest takes place at consolidation, after a period of deferment or forbearance, or at other specific events. For the Income-Contingent Repayment scheme, interest does capitalize annually, but only up to 10\% of the initial loan balance.}

These peculiar features create significant incentives to delay payments: income-driven schemes allow borrowers to minimize present payments and have the balance eventually forgiven. Simple (i.e., non-capitalized) interest implies that the outstanding balance is divided in two: the principal, which accumulates interest at the loan's rate, and accrued interest, which does not produce interest. 
The main countervailing incentive is the loan rate, which is typically higher than the borrower's discount rate and accumulates throughout the life of the loan, thereby leading to a larger tax liability. Also, forgiveness may not be relevant for a small loan, as minimum payments may extinguish its balance before it becomes eligible for forgiveness.

Overall, the complexity of student loan repayment options creates an intricate tradeoff between the benefits of late repayments (forgiveness and simple interest) and their costs (accrued interest), which depend on the loan size relative to income, the loan rate and the borrower's discount rate, the forgiveness horizon, and the income-tax rate. The goal of this paper is to understand this tradeoff and find the optimal repayment strategy for a borrower who wishes to minimize the loan's cost, i.e., the present value of future repayments. 

If the loan is so large (or the income so low) that even maximum payments cannot erode the principal, the optimal strategy is to make minimal payments indefinitely, thereby taking advantage of negative amortization, because simple interest is tantamount to a separate, interest-free loan on all accrued interest, which will also be forgiven. (Put differently, the total balance increases linearly, not exponentially, over time.)

If the loan is so small (or the income so high) that even minimum payments do cover interest, then non-capitalization is irrelevant, and the remaining tension is between forgiveness and compound interest. On one hand, a borrower is tempted to delay repayments until the loan is forgiven and only taxes on the forgiven balance are due. On the other hand, the loan rate is much higher than the borrower's discount rate; hence the cost of delaying payments increases exponentially with the forgiveness horizon, potentially offsetting its ostensible savings. 
If the remaining loan life is short, the benefits of forgiveness override interest costs, hence minimum payments are still optimal; otherwise, maximum payments are optimal up to a \emph{critical horizon}.

The critical horizon is the time at which the benefits of forgiveness equal interest costs: repaying an extra dollar today spares the borrower from paying taxes on its forgiven future value, i.e., principal plus interest. 
If both the loan and the tax rates are high, and forgiveness is far away, then it is best to maximize payments for some time, until the critical horizon, when the remaining life of the loan shortens to the point that the savings from extra payments are null. Then, minimal payments through income-driven schemes are optimal because forgiveness is so close that any additional payment increases the cost of the loan. 

Despite the sheer size of student-loan debt and its proximity to academia, the problem of finding cost-minimizing repayment strategies has not received much attention in the academic literature.
Although it is common knowledge that borrowers with large balances and a relatively low tax rate are better off enrolling in income-driven repayment plans, thereby paying the minimum required by the scheme, typical student resources tend to recommend continued enrollment in such schemes, rather than considering the possibility of a period of maximum payments, which we find to be optimal under certain conditions.

In the United States, student loans were first introduced in the postwar period, and their growth accelerated after the establishment of Sallie Mae in 1973. Student loans are now a mainstream scheme to finance higher education, and the Department of Education estimates that there are nearly 45 million of student debt holders in the US, with 2.5 million borrowers owing more than \$100,000 each.\footnote{ And each of the top hundred borrowers owes more than \$1 million, see \href{https://www.wsj.com/articles/mike-meru-has-1-million-in-student-loans-how-did-that-happen-1527252975}{https://www.wsj.com/articles/mike-meru-has-1-million-in-student-loans-how-did-that-happen-1527252975}.}
While student loans have the merit to expand access to higher education, research in the past decade has also brought potential demerits to light.
Recent empirical work finds that higher balances of student loans contribute to reducing home ownership \cite{mezza2020student}, inhibiting propensity to entrepreneurship \cite{krishnan2019cost} and public sector employment \cite{rothstein2011constrained}, delaying marriage \cite{gicheva2016student}, postponing parenthood \cite{shao2014debt} and enrollment in graduate or professional degrees \cite{malcom2012impact,zhang2013effects}, and increasing the cohabitation with parents \cite{bleemer2014debt,dettling2018returning}.
Similarly controversial is the interaction between student loans and tuition: \cite{lucca2018credit} find that an increase in the subsidized loan maximum leads to a sticker-price increase in tuition of about 60 cents on the dollar, thereby suggesting that colleges (rather than students) may be the beneficiaries of a large fraction of government loan subsidies. (This is the so-called ``Bennett hypothesis'', named after William Bennet, who publicly formulated the link between student loans availability and tuition fees as Secretary of Education in 1987.)

This paper contributes to the understanding of student loans by finding the cheapest repayment strategy that accounts for income-driven repayment and forgiveness. When the loan is very large or very small, we also account for simple interest. For intermediate balances, the complete characterization of the optimal strategy with income-driven repayment, forgiveness, and simple interest remains an open problem that currently requires numerical optimization for specific parameter values.

We focus on the minimization of the cost of the loan because, \emph{a priori}, it is consistent with the maximization of the net worth of a household. \emph{A posteriori}, the cost-minimizing strategy also offers significant protection to negative shocks through income-driven repayment.
In principle, deviating from cost minimization would be justified when an alternative strategy would offer lower risk, but our result suggests that potential improvements in this regard may be rather limited. Indeed, the central risk reduction that can be achieved in student loans is through income-driven repayment schemes, which allow monthly payments to be proportional to income above subsistence, thereby partially hedging income fluctuations. However, our results show that enrollment in such schemes is already optimal for large loan balances (for which the potential risk reduction is largest), for the purpose of minimizing costs -- even neglecting their hedging potential. Put differently, income-driven schemes reduce both costs and risk, which means that the minimization of these two quantities is largely aligned. A partial tradeoff between cost and risk is present when the cheapest strategy entails a period of maximum payments. In this case, the borrower could postpone maximum payments at the price of increased repayment costs over the lifetime of the loan, while only reducing risk for the period by which maximum payments have been postponed.

Our findings rely on some assumptions that make the model tractable and its solution accessible. Our model is intentionally deterministic, for two reasons: First, and foremost, deterministic analysis is more accessible, and is in fact central to students' decision making, as attested by the typical online comparison tools, which do not entertain randomness.\footnote{See 
\href{https://studentloanhero.com/calculators/}{https://studentloanhero.com/calculators/},
\href{https://smartasset.com/student-loans/student-loan-calculator}{https://smartasset.com/student-loans/student-loan-calculator}, and
\href{https://www.calculator.net/student-loan-calculator.html}{https://www.calculator.net/student-loan-calculator.html} among others. }
Second, the scarcity of literature on optimal repayment of student loans suggests that this problem should be first understood in the simplest form that specifies its central elements. This paper offers such an analysis: it extends the results in \cite{GH21}, which considered the effects of income-driven repayment and loan-forgiveness, by including also simple interest, which is an important aspect of loans with large balances. (We are very grateful to Erik Kroll \cite{kroll2021} for bringing this feature to our attention.)

In addition, we do not model certain actions that may be available to borrowers, such as deferment, forbearance, consolidation, delinquency, death, or refinancing through a private loan. Some of these actions are possible only when certain events occur in the borrower's life. Others, such as delinquency, tend to increase student-loan liabilities and reduce access to credit, while private-loan refinancing entails forgoing income-driven repayments and forgiveness. 

The issue of delinquency and default deserves some discussion. Student loans, unlike other unsecured debt such as credit card balances, cannot be discharged in bankruptcy except in very rare circumstances \cite{mueller2019rise}, while borrowers' wages can be garnished for life. As a result, delinquency on student loans does not reduce the borrower's liabilities: instead, it adds collection fees to the loan's balance and significantly reduces access to credit by impairing the debtor's credit score. In addition, a borrower with subsistence income (or no income at all) can remain in good standing \emph{without making payments} by enrolling in income-driven repayment schemes, thereby avoiding delinquency at no cost. 
Empirical work on student loan defaults confirms that defaults are difficult to reconcile with borrowers' optimal choices, and may be due to borrowers' insufficient information about their options \cite{delisle2018federal}. Using individually identifiable information on student loan borrowers, \cite{cornaggia2020mismanages} find that \emph{``the majority of distressed student borrowers have their loans in disadvantageous repayment plans even when eligible for more advantageous options''}. In particular, \cite{looney2019useful} find that over 30\% of student loans of \$5,000 or \emph{less} are in default, even though they would be paid in full in ten years with monthly payments below \$100 (and without income-driven schemes). Also, delinquencies tend to decrease as loan balances increase, contrary to the incentives of strategic default. For these reasons, the model in this paper does not entertain delinquency, as its aim is to identify optimal repayment strategies rather than explain observed defaults. Put differently, our focus is normative rather than positive.

The tax treatment of student loans is also worth mentioning. Unlike mortgage interest, which is fully tax-deductible, student loan interest is deductible up to \$2,500 and only if the borrower's income is below \$70,000 (the deduction is fully phased out above \$85,000).\footnote{See \href{https://www.irs.gov/pub/irs-pdf/p970.pdf}{https://www.irs.gov/pub/irs-pdf/p970.pdf}.} Thus, the tax benefit varies across taxpayers, but for each taxpayer with a minimal loan amount such benefit is virtually constant across strategies, and therefore has no marginal effect in the choice of the optimal strategy. For this reason, the model in this paper does not incorporate the tax benefit explicitly.

Finally, it is worth discussing how the prospect of partial debt cancellation may affect present repayment strategies. The model in the paper does not include this feature explicitly, because the sheer uncertainty on the details of cancellation policies would force any quantitative analysis to rely on arbitrary assumptions. Qualitatively, however, the possibility of debt cancellation is equivalent, \emph{on average} to a lower student-loan rate, with the reduction equal to the probability of cancellation per unit of time, multiplied by the fraction of debt canceled.\footnote{This insight is in analogy to the hazard and recovery rates in the valuation of defaultable bonds.} 
As a result, the prospect of cancellation tends to bolster minimum payments in the early stages of repayment. 
 
The rest of this paper is organized as follows: Section \ref{sec:main} describes in detail the statement of the main result and discusses its quantitative implications. 
Section \ref{sec:simple} discusses the effect of simple accrued interest, obtaining the optimal repayment strategy for very large or small balances, and deriving some necessary conditions for optimality.
Section \ref{sec:proofs} contains the rigorous mathematical proof of the main result in \cite{GH21} -- the case of capitalized interest, which is relevant for small loan balances -- first reducing the search for optimal strategies to a class of max-min strategies, and then identifying the optimal ones within this class. Section \ref{sec:proofsnew} contains new results that account for simple interest, characterizing the optimal repayment strategy for very large or very small balances, and offering two necessary conditions satisfied by optimal strategies. Concluding remarks are in Section \ref{sec:conc}.


\section{Model and Main Result}\label{sec:main}

This section presents the complete solution of the optimal repayment problem in the presence of forgiveness and income-driven schemes. The next section additionally accounts for simple accrued interest.

A student graduates with a loan balance of $x>0$ and seeks the repayment strategy $\alpha$ that minimizes the present value of future payments, discounted at some rate $r>0$, which represents the opportunity cost of money, i.e., the alternative safe return that could be obtained on any dollar used to pay off the loan. For example, a household with a mortgage may ponder whether to increase mortgage or student loan payments, hence the mortgage rate is a close approximation to the household's discount rate. For a household without other debt, the discount rate represents the return on a safe investment. 

The loan carries an interest rate of $r+\beta$, higher than the discount rate (i.e., $\beta>0$), which means that paying off the loan earlier entails lower compounding costs.\footnote{The case of a household with debt that carries a higher interest than student loans, such as credit card debt, is somewhat trivial, as the borrower's optimal policy is to pay off such debt first. Thus, we focus on the case of a positive spread $\beta$.}
Thus, 
denoting by $\alpha_t$ the chosen repayment rate at time $t$, the loan balance $b^\alpha_t$ evolves over time according to the dynamics
\begin{align}\label{eq:1}
db^\alpha_t = (r+\beta)b^\alpha_t dt - \alpha_t dt, \quad b_0 = x>0
.
\end{align}

The student loan also includes a forgiveness provision, whereby at some future horizon $T>0$ the remaining balance $b^\alpha_T$ of the loan is forgiven, but then taxed at rate $\omega\in (0,1)$, whence a payment of $\omega b^\alpha_T$ is due at time $T$. Such a provision encourages delaying payments as the forgiveness horizon approaches, thereby countering the compounding motive. 

The payment rate at time $t$ is constrained to the range $m(t)$ to $M(t)$, which depends on the former student's income, with $m(t)$ reflecting the minimum payment due under income-driven repayments, and $M(t)$ the maximum payment that accommodates other living expenses without incurring debt, such as credit card balances, which carry a higher rate than that of student loans.
Specifically, for any $x>0$ and Lebesgue measurable $\alpha:[0,T]\to [0,\infty)$, the present value of future payments is 
\begin{equation}\label{J}
J(x,\alpha) := \int_{0}^{\tau} e^{-rt} \alpha_t dt + e^{-r\tau} \omega b_{\tau},
\end{equation}
where 
\begin{align}\label{tau}
\tau:= \inf\{t \geq 0 : b_t = 0\} \wedge T
\end{align}
is the time when the loan is either paid in full or forgiven. The goal is to minimize the present value of future payments, i.e.,
\begin{equation}\label{v}
v(x) := \inf_{\alpha\in \mathcal A} J(x, \alpha), 
\end{equation}
where the set of admissible repayment strategies is defined as
\[
\mathcal A := \{ \alpha:\text{$t\mapsto \alpha_t$ is Lebesgue measurable with } m(t) \leq \alpha_t \leq M(t)\ \hbox{for}\ 0\le t\le T\},
\]
for some Lebesgue integrable $m,M:[0,T]\to(0,\infty)$ satisfying $m(t)<M(t)$ for all $t\in[0,T]$.
The main result describes the optimal repayment strategy in relation to the loan's parameters:
\begin{theorem}\label{th:main}
Define the critical horizon as
\begin{equation}\label{eq:critical}
t_\oc := \left(T + \frac{\log {\omega}}{\beta}\right)^+\in[0,T)
,
\end{equation}
and let
$
x^* := \int_{0}^{t^*} e^{-(r+\beta)s}  M(s) ds>0,
$
where the time $t^*\in (t_\oc,T)$ is the unique solution to 
\begin{equation}\label{t_2^*}
\int_{t_\oc}^{t^*} e^{-rs}  M(s) (1-\omega e^{\beta(T-s)}) ds = \int_{t_\oc}^{T} e^{-rs}  m(s)  (1-\omega e^{\beta(T-s)}) ds
.
\end{equation}
Then, for any $x>0$, the strategy $\alpha^*\in\mathcal A$ defined as
\[
 \alpha^*_t := 
 \begin{cases}
 M(t) 1_{[0,t_\oc]}(t)+m(t) 1_{(t_\oc,T]}(t) & {t\in[0,T],\quad \hbox{if}\ x>  x^*},\qquad \text{(max-min)}\\
 M(t)\hspace{1.65in} & {t\in[0,T],\quad \hbox{if}\ x\le  x^*},\qquad \text{(max)}\\
 \end{cases}
\]
attains the minimum loan value. Also, $v(x) = v_1(x)$ for $x>  x^*$ and $v(x) = v_2(x)$ for $x\le x^*$, where
\begin{align}
v_1(x) &:= \!\!\int_{0}^{t_\oc}\!\!\!\!  e^{-rs} M(s) ds+ \!\!\int_{t_\oc}^{T}\!\!\!\!  e^{-rs}  m(s) ds \! + \omega e^{\beta T}\left(x\!\!  -\!\!  \int_{0}^{t_\oc}\!\!\!\!  e^{-(r+\beta)s}  M(s) ds \! - \!\!\int_{t_\oc}^{T}\!\!\!\!  e^{-(r+\beta)s}  m(s) ds \right), \label{v_1}\\
v_2(x) &:=  \int_{0}^{t_M}\!\!\!\!  e^{-rs}  M(s) ds,
\quad\text{where $t_M>0$ satisfies }\ x=  \int_{0}^{t_M} e^{-(r+\beta)s}  M(s) ds.\label{v_2}
\end{align}
\end{theorem}

\begin{figure}
\centering
\includegraphics[width=.49\textwidth]{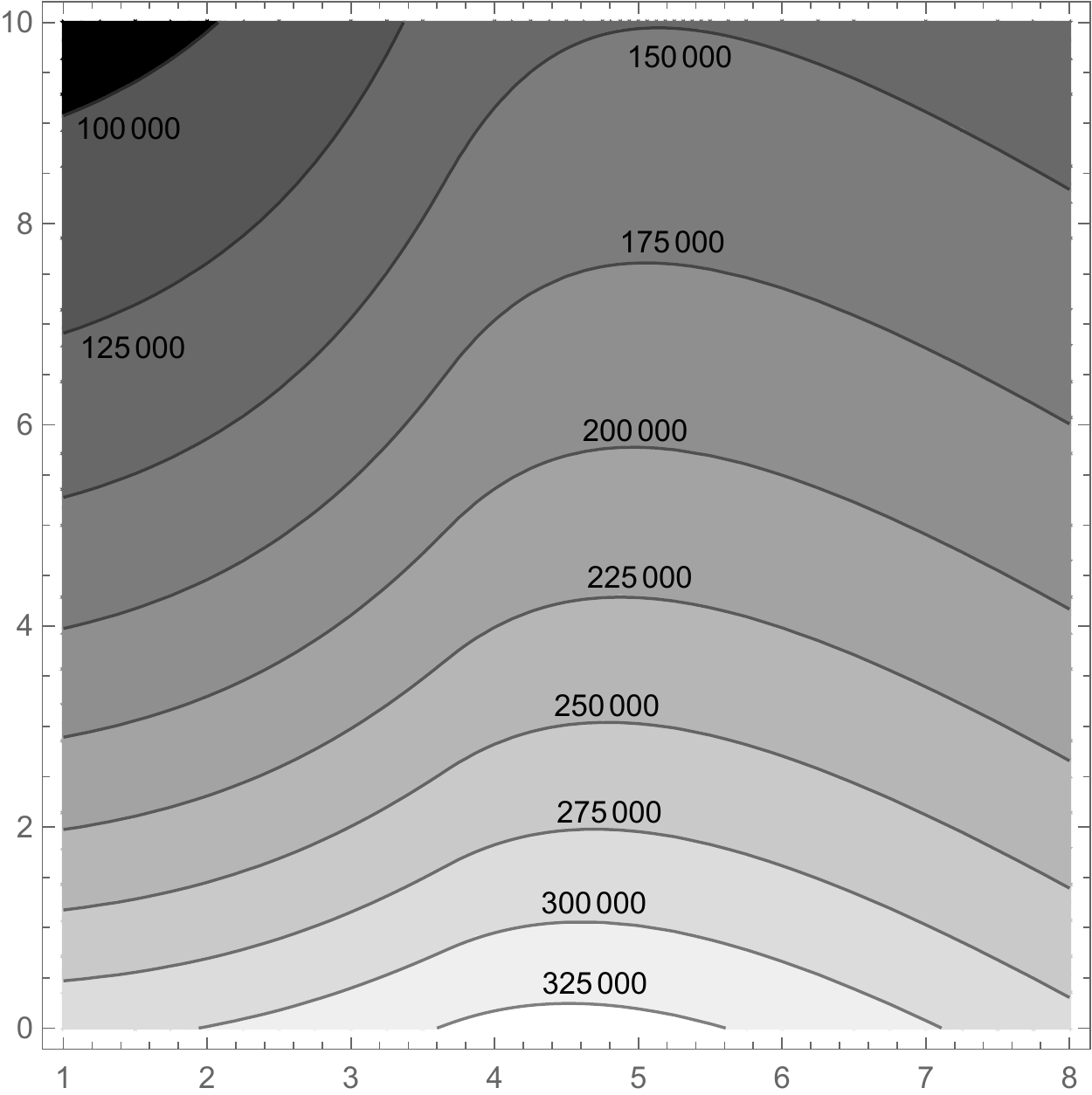}
\includegraphics[width=.49\textwidth]{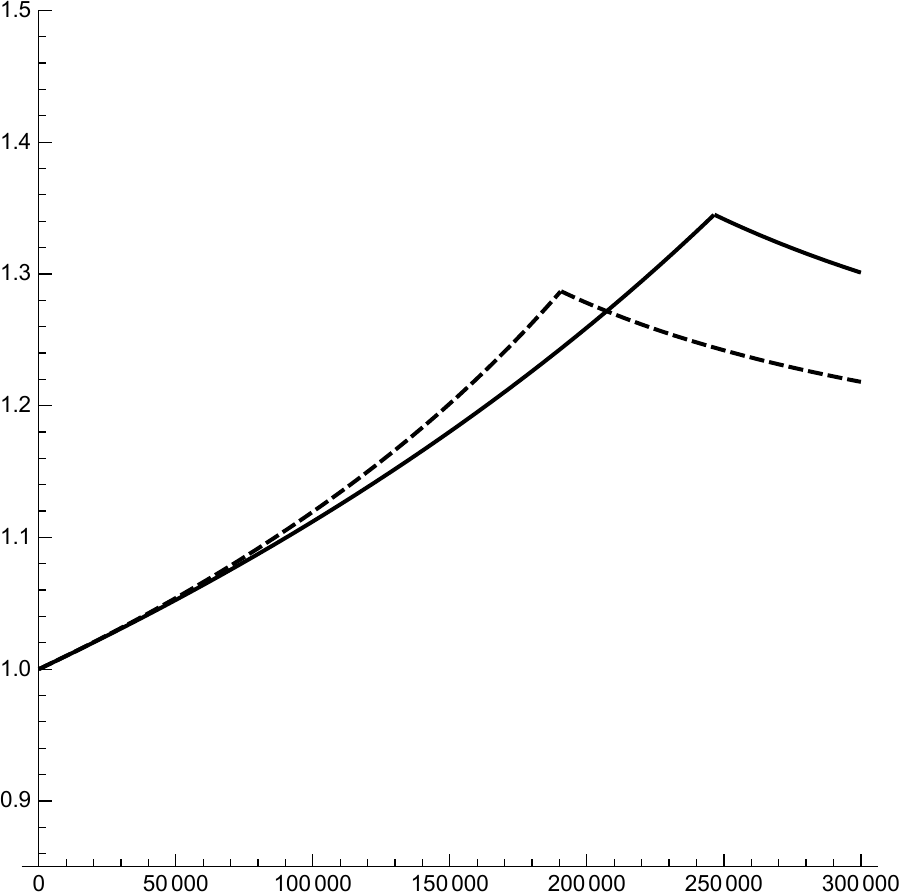}
\caption{\label{fig:xs}
Left: Critical balance $x^*$ (contours), above which the max-min strategy is cheaper than the max repayment strategy, against loan spread $\beta$ (horizontal, in percent) and discount rate (vertical, in percent). 
Right: Cost-to-balance ratio (vertical) against loan balance (horizontal) for PLUS loans (7.54\% rate), with discount rate of 3\% (solid) and 6\% (dashed).
Parameters: forgiveness horizon $T=25$, annual growth of income and poverty level $g=4\%$, tax rate $\omega=40\%$, minimum and maximum payments are 10\% and 30\% of income above subsistence of \$32,000.
}
\end{figure}

The message of this result is straightforward: the cheapest repayment strategy mandates maximum payments when the initial balance is sufficiently low ($x<x^*$, ``max'' strategy). Otherwise ($x>x^*$, ``max-min'' strategy), maximum payments are in order before the critical horizon $t_\oc$ in \eqref{eq:critical}, at which point enrollment in the income-driven scheme takes place, implying minimum payments thereafter. If the critical horizon is zero (for example, if either the tax rate or the interest rate spread are very low), then enrollment is immediate, and minimum payments span the entire life of the loan (i.e., the ``max-min'' boils down to ``min''). The {\it critical balance} $x^*$ that separates these two regimes is precisely the unique balance which yields the same repayment cost under both strategies.

The left panel of figure \ref{fig:xs} displays the critical balance $x^*$ as a function of the borrower's discount rate $r$ and the loan spread $\beta$, and shows that its dependence on these rates is highly nonlinear. The critical balance is particularly sensitive to the discount rate $r$, with low rates making it optimal to repay large balances early, and high discount rates encouraging deferral. The intuition is clear: a borrower with a higher opportunity-cost of capital has a stronger preference to later payments because they entail a lower sacrifice in return. 

Note also that the ostensible complexity of calculating the critical balance $x^*$ is not a significant barrier for a borrower who wishes to choose the cheapest repayment strategy: in practice, the borrower only needs to compare the cost of the max and max-min strategies, choosing the cheaper one among these two. An important corollary of this result is that only large loan balances, those above $x^*$, benefit from income-driven repayment schemes. Instead, smaller balances should be paid off as early as possible through maximum payments. 

To better understand this issue, the right panel of figure \ref{fig:xs} plots the cost-to-balance ratio for PLUS loans for discount rates $3\%$ and $6\%$, representative of borrowers with different credit scores.
Discount rates have a minor impact on the valuation of small loan balances, leading to noticeable differences only after enrollment in income-driven schemes becomes optimal. Indeed, a higher discount rate lowers the enrollment threshold $x^*$, significantly decreasing the borrowing cost per unit of balance. Once such threshold is exceeded, the marginal cost of any additional borrowed dollar is exactly $ \omega e^{\beta T}$ and the additional balance neither affects payments in the ``max'' nor in the ``min'' periods of the loan. 

In summary, the marginal cost of borrowing increases with the balance until enrollment in income-driven repayment becomes optimal. At that point, the marginal cost of additional borrowing drops to the constant $\omega e^{\beta T}$, as the average cost of borrowing gradually declines to the same constant. Thus, an implication of income-driven repayment schemes is that the average unit cost of  borrowing is higher for medium balances than it is for very high balances.
As the next section shows, this conclusion is made even stronger by the practice of non-capitalization of interest, which further reduces the cost of borrowing for large balances.

\section{Simple Interest}\label{sec:simple}
The non-capitalization of interest, which is typical of most student loans in income-driven repayment plans, introduces an additional complication for the optimization problem, as it requires keeping track of the amount of unpaid principal $p^\alpha_t$ in addition to the total balance $b^\alpha_t$ (i.e., principal plus accrued interest). Because the principal can be repaid only after any accrued interest has been paid off, the dynamics of the loan is described by:
\begin{align}
d b^\alpha_t &= ((r + \beta) p^\alpha_t - \alpha_t)dt,\quad b_0 = x>0;\label{b}\\
p^\alpha_t &= \inf_{0\le s\le t} b^\alpha_s.\label{p}
\end{align}
The first equation states that the loan balance increases by the loan rate times the remaining principal, as prescribed by simple interest with negative amortization, and decreases at the current repayment rate. The second equation identifies current principal as the running minimum of the loan balance, reflecting the priority of accrued interest relative to the principal.

To understand the impact of simple interest, it is useful to introduce the following definition.
For any $x>0$ and a strategy $\alpha\in\A$, define the first time of principal repayment as
\begin{equation}\label{t_alpha}
\theta(\alpha) := \inf\left\{t \in[0,T] : p^\alpha_t < x\right\}
,
\end{equation}
following the convention that $\inf\emptyset =T$. For $\alpha_s = M(s)$ (resp.\ $m(s)$) for all $s\in [0,T]$, the corresponding times are denoted by $\theta(M)$ (resp. $\theta(m)$) for $\theta(\alpha)$. Hence, \eqref{b}-\eqref{p} imply that
\begin{equation}\label{p=x}
p^\alpha_t = x\quad \hbox{for}\ t\in [0,\theta(\alpha)].
\end{equation}
That is, an admissible strategy $\alpha\in\A$ repays only the interest portion of the loan by time $\theta(\alpha)$. At time $\theta(\alpha)$, all previously unpaid interest has finally been paid off and the principal portion of the loan is next in line for repayment.

The next Lemma shows how to improve a strategy, depending on the value of $\theta(\alpha)$:
\begin{lemma}\label{lem:reduce to B}
For any $x>0$ and $\alpha\in\A$, 
\begin{itemize}
\item [(i)] 
if $\theta(\alpha) =T$, then $J(x,m) \le J(x,\alpha)$.
\item [(ii)] 
if $\theta(\alpha) <T$, then there exists a unique $t_0\in [0,\theta(\alpha)]$ that satisfies
\begin{equation}\label{payment equal}
\int_{0}^{\theta(\alpha)}  \alpha_s ds = \int_{0}^{t_0} m(s) ds + \int_{t_0}^{\theta(\alpha)} M(s)ds.
\end{equation}
Moreover, $\overline\alpha\in\A$ defined by 
\begin{equation}\label{mMalpha}
\overline\alpha_t := m(t) 1_{[0,t_0]}(t) + M(t) 1_{(t_0,\theta(\alpha)]}+\alpha_t 1_{(\theta(\alpha),T]}(t),\quad \forall 0\le t\le T,  
\end{equation}
satisfies $\theta({\overline\alpha}) = \theta(\alpha)$ and $J(x,\overline\alpha)\le J(x,\alpha)$.
\end{itemize}
\end{lemma}
Part (ii) implies that any strategy that starts with negative amortization (i.e., for which $\theta(\alpha)>0$) can be improved by first minimizing and then maximizing payments before the principal is repaid. Furthermore, part (i) stipulates that, if the principal is never repaid before forgiveness, then permanent minimal payments reduce costs. As an immediate corollary, 
if a loan is so large that even maximal payments would never erode the principal, then minimal payments are optimal.

\begin{corollary}\label{cor:verylarge}
For any $x>0$, 
\begin{itemize}
\item [(i)] if $\theta(M)=T$, then $\alpha^*_t := m(t)$, $t\in[0,T]$ is optimal for \eqref{v}. 
\item [(ii)] if $\theta(m)=0$ and $m(t)$ is nondecreasing, then $\alpha^*$ in Theorem~\ref{th:main} is optimal for \eqref{v}.
\end{itemize}
\end{corollary}
The above result is relevant for those large loans, typical of graduate or professional degrees, which (i) carry a high interest rate, and (ii) do not result in immediate high earnings upon graduation. For such loans, the first few years necessarily result in negative amortization, accumulating a large balance of accrued interest that must be repaid before the principal. In such cases, even if subsequent increases in income would allow maximal payments to exceed interest, unless past accrued interest can also be repaid before forgiveness, then it is optimal to keep payments to the minimum for the entire life of the loan. (Hence such loans should be valued accordingly.)

Instead, when even the minimum payments of an income-driven scheme generates positive amortization, the analysis in the previous section applies because $p^\alpha_t = b^\alpha_t$ for all $\alpha\in\mathcal A$ and $t\in[0,T]$.

Lemma~\ref{lem:reduce to B} improves a strategy $\alpha\in\A$ by optimizing the repayment of accrued interest. The next Lemma, on the other hand, examines how to more efficiently repay the principal portion of the loan. Mathematically, it offers an additional necessary condition for optimality by stipulating that, when a strategy is repaying principal over a time interval $[a,c]$, on that interval maximum payments should always precede minimum payments.
\begin{lemma}\label{lem:p decreases}
Fix any $x>0$ and $\alpha\in\A$ with $\theta(\alpha) <T$. Suppose that there exist $a,c\in[\theta(\alpha),T]$ with $a<c$ such that $t\mapsto p^\alpha_t$ is strictly decreasing on $[a,c]$. If $\alpha\in\A$ does not belong to the collection
\[
\B_{[a,c]} := \{\alpha\in\A : \exists s_0\in [a, c]\ \hbox{s.t.}\ \alpha_t = M(t) 1_{[a, s_0]}+ m(t) 1_{(s_0,c]}(t)\ \hbox{for a.e.}\ t\in[a,c]\}, 
\]
then there exists $u\in(a,c)$ such that $\alpha_{(u)}\in\A$ defined by
\begin{equation}\label{alMmal}
(\alpha_{(u)})_t := \alpha_t 1_{[0,a]}(t) + M(t) 1_{(a,u]}(t) +m(t) 1_{(u,c]}(t) +\alpha_t 1_{(c,T]}(t) \quad \forall t\in[0,T]
\end{equation}
satisfies $J(x,\alpha_{(u)})< J(x,\alpha)$.
\end{lemma}

This lemma implies that, when an optimal strategy is in positive amortization (i.e., it is repaying principal), then payments should be first maximal and then minimal. Otherwise, it can be improved by altering repayment rates so that they are initially maximal and then minimal.
In particular, repaying a loan through a constant repayment rate cannot be optimal, unless such a constant happens to be the maximum payment that a borrower can afford. As constant repayments are the default for student-loans, this observation indicates that inaction is unlikely to be optimal for any borrower.

Finally, it is worth highlighting the stark difference that simple interest makes in the marginal cost of very large loans. As observed above, once a loan becomes large enough, the repayment strategy does not change, and the only marginal effect of a higher balance is through the increased tax liability at the forgiveness horizon. In the case of compound interest (Section~\ref{sec:main}), the present value of an additional dollar is then $\omega e^{-r T} e^{(r+\beta)T} = \omega e^{\beta T}$, which is insensitive to the discount rate $r$.

By contrast, in the case of simple interest the marginal cost of an additional dollar borrowed becomes $\omega e^{-r T} (1+(r+\beta)T)$, which reflects the present value of linearly increasing accrued interest, and depends on both the discount rate and the loan spread separately. The difference between these formulas can be very significant: for example, with a $r=3\%, \beta=4\%, \omega = 40\%, T = 25$, the
marginal cost is \$1.09 with compound interest, but it decreases to \$0.52 with simple interest. As the discount rate increases, simple interest becomes even more advantageous, making additional balances increasingly inconsequential.

\section{Proofs for Section~\ref{sec:main}}\label{sec:proofs}

This section contains the proof of Theorem~\ref{th:main} (the main result in \cite{GH21}), which identifies the cheapest repayment strategy in relation to the initial balance in the case of compound interest. 

First, Lemma \ref{pro:2.1} reduces the search for the optimal strategy to the class of strategies with maximum, followed by minimum payments (with the latter possibly absent). 
Next, Propositions~\ref{prop:x large}, \ref{prop:x small}, and \ref{prop:x intermediate} altogether demonstrate that the optimal strategy must be either $\alpha^{1}_t := M(t) 1_{[0,t_\oc]}(t)+m(t) 1_{(t_\oc,T]}(t)$ or $\alpha^{2}_t := M(t)$. Finally, Lemma~\ref{lem:v1-v2} compares the costs of $\alpha^1$ and $\alpha^2$, establishing Theorem~\ref{th:main} at the end of this section. 
The discussion begins by observing a simple expression for the remaining balance \eqref{eq:1} in terms of the initial balance and the discounted value of repayments.
\begin{remark}\label{thm:1.1}
For any measurable $\alpha:[0,T]\to[0,\infty)$, the unique solution to \eqref{eq:1} is 
\begin{equation}\label{b formula}
b_t= e^{(r+\beta)t}\left(x-\int_{0}^{t} e^{-(r+\beta)s} \alpha_s ds\right),\quad t\ge 0.
\end{equation}
Indeed, the claim follows by integrating the equality
\begin{align*}
d(e^{-(r+\beta)s} b_s) 
=e^{-(r+\beta)s} \{ -(r+\beta) b_s ds + db_s \}  = -e^{-(r+\beta)s} \alpha_s ds.
\end{align*}
\end{remark}


The next result shows that it is sufficient to consider repayment strategies of a very specific form: repaying first at the maximum rate $M(t)$ and then at the minimum rate $m(t)$. 
\begin{lemma}\label{pro:2.1}
For any $x>0$, 
$
v(x) =  \inf_{\alpha \in \mathcal B} J(x,\alpha),
$
 where 
\begin{align}\label{B}
\mathcal B &:= \{ \alpha\in\mathcal A: \exists\ t_0\ge 0\ \hbox{s.t.}\ \alpha_t= M(t) 1_{[0,t_0]}(t)+ m(t)1_{(t_0, T]}(t)\  \hbox{for a.e.}\ t\in[0,T]\}.
\end{align}
\end{lemma}

\begin{proof}
Fix $x>0$. First, observe that
$
\inf_{\alpha \in \mathcal B} J(x,\alpha)=\inf_{\alpha \in \mathcal B'} J(x,\alpha),
$
where 
\[
\mathcal B' := \{ \alpha\in\mathcal A: \exists\ t_0\ge 0\ \hbox{s.t.}\ \alpha_t= M(t) 1_{[0,t_0]}(t)+ m(t)1_{(t_0, \tau]}(t)\ \hbox{for a.e.}\ t\in[0,\tau]\}.
\]
Indeed, for any $\alpha\in\mathcal B$, the truncated strategy $\alpha'$, defined by $\alpha'_t:= \alpha_t 1_{[0,\tau]}(t)$, belongs to $\mathcal B'$ and satisfies $J(x,\alpha')=J(x,\alpha)$; conversely, for any $\alpha'\in\mathcal B'$, the extended strategy $\alpha$, defined by $\alpha_t:= \alpha'_t 1_{[0,\tau]}(t)+ m 1_{(\tau,T]}(t)$, belongs to $\mathcal A$ and satisfies $J(x,\alpha)=J(x,\alpha')$. For this reason, the remaining proof focuses on establishing $v(x)=\inf_{\alpha \in \mathcal B'} J(x,\alpha)$. To this end, it remains to show that for any $\alpha \in \mathcal A\setminus \mathcal B'$, there exists $\bar\alpha \in \mathcal B'$ such that $J(x,\bar\alpha)<J(x,\alpha)$, i.e.
\begin{equation}\label{toshow}
\int_{0}^{\tau(\bar\alpha)} e^{-rt} \bar\alpha_t dt + \omega e^{-r\tau(\bar\alpha)}b^{\bar\alpha}_{\tau(\bar\alpha)} < \int_{0}^{\tau(\alpha)} e^{-rt} \alpha_t dt + \omega e^{-r\tau(\alpha)} b^\alpha_{\tau(\alpha)},
\end{equation}
where $\tau$ in \eqref{tau} is denoted as $\tau(\bar\alpha)$ or $\tau(\alpha)$, and $b$ in \eqref{eq:1} as $b^{\bar\alpha}$ or $b^\alpha$, to emphasize their dependence on the chosen repayment strategy. 

For any $\alpha \in \mathcal A\setminus \mathcal B'$, the first claim is that there exists $0<t_0<\tau(\alpha)$ such that 
\begin{equation}\label{t_0}
\int_{0}^{t_0} (M(t)-\alpha_t)e^{-(r+\beta)t}  dt = \int_{t_0}^{\tau(\alpha)}(\alpha_t-m(t))e^{-(r+\beta)t} dt.
\end{equation}
Define $f:[0,\tau(\alpha)] \to \R$ by
$
f(t) := \int_{0}^{t} (M(s)-\alpha_s)e^{-(r+\beta)s}  ds - \int_{t}^{\tau(\alpha)}(\alpha_s-m(s))e^{-(r+\beta)s} ds. 
$
{} As $\alpha, M, m$ are all Lebesgue integrable, $f$ is by definition continuous. Also, $\alpha \notin \mathcal B'$ implies that $\alpha_t$ cannot be equal to $m(t)$ for a.e. $t\in[0,\tau(\alpha)]$, whence $f(0) =  - \int_{0}^{\tau(\alpha)}(\alpha_s-m(s))e^{-(r+\beta)s} ds<0$. Likewise, $\alpha_t$ cannot be equal to $M(t)$ for a.e. $t\in[0,\tau(\alpha)]$, implying $f(\tau(\alpha)) = \int_{0}^{\tau(\alpha)} (M(s)-\alpha_s)e^{-(r+\beta)s}  ds>0$. The continuity of $f$ thus ensures the existence of $0<t_0<\tau(\alpha)$ such that $f(t_0)=0$, i.e. \eqref{t_0} holds. 
Now, define $\bar\alpha:[0,T]\to[0,\infty)$ by
\begin{equation}\label{bar alpha}
  \bar \alpha_t := M(t)1_{[0,t_0]}(t) +m(t)1_{(t_0,\tau(\alpha)]}(t),\quad 0\le t\le T.  
\end{equation}
Observe that 
$
\tau(\bar\alpha)=\tau(\alpha)
$
. Indeed, by \eqref{t_0}, 
\begin{multline*}
\int_{0}^{\tau(\alpha)} e^{-(r+\beta)t} \bar \alpha_t dt = \int_{0}^{\tau(\alpha)} e^{-(r+\beta)t}  \alpha_t dt + \int_{0}^{t_0}(M(t)-\alpha_t) e^{-(r+\beta)t}  dt \\
 -  \int_{t_0}^{\tau(\alpha)}(\alpha_t-m(t))e^{-(r+\beta)t} dt = \int_{0}^{\tau(\alpha)} e^{-(r+\beta)t}  \alpha_tdt.
\end{multline*}
This fact, together with Remark~\ref{thm:1.1}, implies 
$
b^{\bar \alpha}_{\tau(\alpha)} = b^\alpha_{\tau(\alpha)}. 
$
{} The case $b^\alpha_{\tau(\alpha)}>0$ leads to $\tau(\alpha)=T$ and thus $b^{\bar \alpha}_{T}=b^{\bar \alpha}_{\tau(\alpha)}>0$, which readily implies $\tau(\bar\alpha)= T = \tau(\alpha)$. If $b^\alpha_{\tau(\alpha)}=0$, then $b^{\bar \alpha}_{\tau(\alpha)} = 0$ and thus $\tau(\bar\alpha)\le \tau(\alpha)$, thanks to the definition of $\tau$ in \eqref{tau}. If $\tau(\bar\alpha)< \tau(\alpha)\le T$, then $b^{\bar \alpha}_{\tau(\bar\alpha)} = 0$, again by \eqref{tau}. It then follows from the definition of $\bar\alpha$ and the formula of $b^{\bar \alpha}$ in \eqref{b formula} that $b^{\bar\alpha}_{\tau(\alpha)}< b^{\bar\alpha}_{\tau(\bar\alpha)}=0$, a contradiction. Thus, $\tau(\bar\alpha)=\tau(\alpha)$ as required, which 
implies $\bar\alpha\in\mathcal B$. 

It remains to show \eqref{toshow}. As a consequence of \eqref{t_0},
\begin{align}\label{<}
 e^{-\beta t_0} \int_{0}^{t_0} e^{-rt}(M(t)-\alpha_t)dt &< \int_{0}^{t_0} e^{-(r+\beta)t} (M(t)-\alpha_t)dt \nonumber\\
 &=\int_{t_0}^{\tau}e^{-(r+\beta)t}(\alpha_t-m(t))dt < e^{-\beta t_0}\int_{t_0}^{\tau}e^{-rt}(\alpha_t-m(t)) dt. 
\end{align}
It follows that  
\begin{multline*}
\int_{0}^{\tau(\bar\alpha)} e^{-rt} \bar \alpha_t dt +  \omega e^{-r\tau(\bar\alpha)}b_{\tau(\bar\alpha)}^{\bar \alpha} = \int_{0}^{\tau(\alpha)} e^{-rt}  \alpha_t dt +\int_{0}^{t_0} e^{-rt} (M(t) -\alpha_t) dt \\ 
-  \int_{t_0}^{\tau(\alpha)} e^{-rt} (\alpha_t-m(t))dt 
+ \omega e^{-r\tau(\alpha)}b_{\tau(\alpha)}^{\bar\alpha}
< \int_{0}^{\tau(\alpha)} e^{-rt}  \alpha_t dt + \omega e^{-r\tau(\alpha)}b_{\tau(\alpha)}^{\alpha},
\end{multline*}
where the equality follows from $\tau(\bar\alpha)=\tau(\alpha)$  and the definition of $\bar\alpha$ in \eqref{bar alpha}, and the inequality is due to \eqref{<} and $b^{\bar \alpha}_{\tau(\alpha)} = b^\alpha_{\tau(\alpha)}$. That is, \eqref{toshow} is established. 
\end{proof}


\subsection{Three Cases}
The following analysis distinguishes three cases, depending on how large the initial balance of the loan is. Consider the two useful thresholds
\begin{equation}\label{x min, max}
\underline{x}:= \int_{0}^{T} e^{-(r+\beta)s}  m(s) ds\quad \hbox{and}\quad \overline{x}:= \int_{0}^{T} e^{-(r+\beta)s}  M(s) ds. 
\end{equation}

The first case is that of an initial balance $x>0$ of the loan so large that even maximum payments cannot pay it off by time $T$. 

\begin{proposition}\label{prop:x large}
Fix $x> \overline x$ and recall $t_\oc\in [0,T)$ defined in \eqref{eq:critical}.
Then, $\alpha^*\in\mathcal A$ defined by
\begin{equation}\label{alpha^* x large}
 \alpha^*_t := M(t) 1_{[0,t_\oc]}(t)+m(t) 1_{(t_\oc,T]}(t),\quad 0\le t\le T, 
\end{equation}
is an optimal control. Moreover, $\tau(\alpha^*)=T$ and $v(x) = v_1(x)$, with $v_1$ defined as in \eqref{v_1}.
\end{proposition}

\begin{proof}
Note that with $x> \overline x$ even maximum payments, i.e. $\tilde\alpha_t:= M(t)$ for all $0\le t\le T$, cannot pay off the debt by time $T$. Indeed, 
$
b^{\tilde\alpha}_T= 
e^{(r+\beta)T}\left(x-\int_{0}^{T}e^{-(r+\beta)s} M(s) ds\right) >0
$
{} by \eqref{b formula}, whence 
$
b^\alpha_T>0 \hbox{ and } \tau(\alpha)=T, \hbox{ for all } \alpha\in\mathcal A. 
$
{} Thus, a strategy of the form 
\begin{equation}\label{alpha^*}
   \alpha^*_t := M(t) 1_{[0,t_0]}(t)+m(t) 1_{(t_0,T]}(t),\quad 0\le t\le T, 0\le t_0\le T, 
\end{equation} 
satisfies
\begin{gather}\label{=f(t^*)}
J(x,\alpha^*) = \int_{0}^{T} e^{-rt} \alpha^*_t dt + \omega e^{-rT} b^{\alpha^*}_T = f(t_0), 
\quad\text{where}\\
\label{f}
f(t) := \int_{0}^{t} e^{-rs}  M(s) ds+ \int_{t}^{T} e^{-rs}  m(s) ds
+ \omega e^{\beta T}\left(x- \int_{0}^{t} e^{-(r+\beta)s}  M(s) ds  -\int_{t}^{T} e^{-(r+\beta)s}  m(s) ds \right).
\end{gather}
Note that the second equality in \eqref{=f(t^*)} follows from \eqref{alpha^*} and Remark~\ref{thm:1.1}. By direct calculation,
$
f'(t) = e^{-r t}(M(t)- m(t)) \left(1 - \omega e^{\beta (T-t)} \right),
$
{} which shows that $f$ is strictly decreasing for $t< T+ \frac{\ln \omega}{\beta}$ and strictly increasing for $t>T+ \frac{\ln \omega}{\beta}$. It then follows from \eqref{=f(t^*)} that by taking $t_0 = t_\oc$ in \eqref{eq:critical}, $\alpha^*$ in \eqref{alpha^*} attains $\inf_{\alpha\in\mathcal B}J(x,\alpha)=v(x)$, where the equality follows from Lemma \ref{pro:2.1}. 
\end{proof}

Next, consider the case where the initial balance $x>0$ of the loan is so small that even minimum payments can pay it off by time $T$. 

\begin{proposition}\label{prop:x small}
Fix $0<x\le \underline x$. 
Consider the unique $t_M\in (0,T]$ such that 
\begin{equation}\label{t_2}
x=  \int_{0}^{t_M} e^{-(r+\beta)s}  M(s) ds.
\end{equation}
Then, $\alpha^*\in\mathcal A$ defined by $\alpha^*_t=M(t)$, $0\le t\le T$, 
is an optimal control. Moreover, $\tau(\alpha^*) = t_M<T$ and $v(x)= v_2(x)$, with $v_2$ defined as in \eqref{v_2}. 
\end{proposition}

\begin{proof}
As $0<x\le \underline x$, even minimum payments ($\tilde\alpha_t:= m(t)$ for all $0\le t\le T$) pay off the debt by time $T$. Indeed, 
$
b^{\tilde\alpha}_T = 
e^{(r+\beta)T}\left(x-\int_{0}^{T}e^{-(r+\beta)s}m(s)  ds\right) \le 0
$
{} by \eqref{b formula}, whence
\begin{equation}\label{b=0}
b^\alpha_\tau = 0\quad \hbox{and}\quad \tau(\alpha)\le T,\qquad \hbox{for all $\alpha\in\mathcal A$.} 
\end{equation}
Also, observe that $0<x\le  \underline x$ and $0<m(t)<M(t)$ readily imply the existence of a unique $t_M\in(0,T]$ such that \eqref{t_2} holds. 


Now, focus on strategies $\alpha^*$ as in \eqref{alpha^*}, with $0\le t_0\le T$. For each $0\le t_0\le T$, since $b^{\alpha^*}_\tau = 0$ by \eqref{b=0}, it follows that $x = \int_{0}^{\tau} e^{-(r + \beta)s} \alpha^*_s ds$, in view of \eqref{b formula}. This fact, together with \eqref{alpha^*} and \eqref{t_2}, implies
\begin{align}\label{x=int}
x &=  \int_{0}^{t_0} e^{-(r+\beta)s}  M(s) ds +  \int_{t_0}^{\tau} e^{-(r+\beta)s}  m(s) ds,\quad \hbox{for}\ 0\le t_0\le t_M. 
\end{align}
Thus, $\tau$ is a function $t_0\mapsto\tau(t_0)$, $0\le t_0\le t_M$. By the strict positivity of $M$ and $m$, \eqref{x=int} indicates that $\tau(t_0)=t_0 + \eta(t_0)$, where
\begin{equation}\label{tau>t_0}
t_0\mapsto \eta(t_0)\ \ \hbox{is strictly decreasing on $[0,t_M]$ with $\eta(t_M)=0$}. 
\end{equation}
It then follows from the decreasing property of $\eta$ and \eqref{x=int} that $t_0\mapsto\tau(t_0)$ is differentiable a.e. Indeed, given $0\le t_0\le t_M$, for any $h\in\R$ such that $0<t_0+h<t_M$, \eqref{x=int} entails
\[
\int_{0}^{t_0}\!\!\!\! e^{-(r+\beta)s}  M(s) ds +  \int_{t_0}^{\tau(t_0)}\!\!\!\! e^{-(r+\beta)s}  m(s) ds=\int_{0}^{t_0+h}\!\!\!\! e^{-(r+\beta)s}  M(s) ds +  \int_{t_0+h}^{\tau(t_0+h)} \!\!\!\! e^{-(r+\beta)s}  m(s) ds, 
\]
which reduces to
\begin{equation}\label{reduce to}
\int_{\tau(t_0+h)}^{\tau(t_0)} e^{-(r+\beta)s}  m(s) ds = \int_{t_0}^{t_0+h} e^{-(r+\beta)s}(M(s)-m(s)) ds.
\end{equation}
Thus, the right-hand side above vanishes as $h\to 0$, hence $\tau(t_0+h)\to \tau(t_0)$, i.e. $t_0\mapsto\tau(t_0)$ is continuous and so is $t_0\mapsto\eta(t_0)$. By the continuity of $\eta$, the Lebesgue differentiation theorem implies that, for a.e. $t_0\in [0,t_M]$,
\begin{align*}
\lim_{h\to 0} \frac{1}{\tau(t_0)-\tau(t_0+h)} \int_{\tau(t_0+h)}^{\tau(t_0)} e^{-(r+\beta)s}  m(s) ds =& e^{-(r+\beta)\tau(t_0)}  m(\tau(t_0)),\\
\lim_{h\to 0} \frac{1}{h} \int_{t_0}^{t_0+h} e^{-(r+\beta)s}  (M(s)-m(s)) ds =& e^{-(r+\beta)t_0}  (M(t_0)-m(t_0)).
\end{align*}
Dividing the second equality above by the first one and recalling \eqref{reduce to} yields
\begin{equation}\label{tau'}
\tau'(t_0) = \lim_{h\to\ 0} \frac{\tau(t_0+h)-\tau(t_0)}{h}= - e^{(r+\beta)(\tau(t_0)-t_0)}\frac{M(t_0)-m(t_0)}{m(\tau(t_0))}. 
\end{equation}
Thanks to \eqref{alpha^*} and \eqref{b=0}
\begin{align}\label{=g}
J(x,\alpha^*) = \int_{0}^{\tau} e^{-rt} \alpha^*_t dt + \omega e^{-r\tau} b^{\alpha^*}_\tau  =  \int_{0}^{\tau} e^{-rt} \alpha^*_t dt=g(t_0),
\end{align}
where $g:[0,t_M]\to\R$ is defined as
\begin{equation}\label{g}
g(t_0) :=   \int_{0}^{t_0} e^{-rs}  M(s) ds +  \int_{t_0}^{\tau(t_0)} e^{-rs}  m(s) ds.
\end{equation}
By direct calculation, 
\begin{align}\label{g'}
g'(t_0) &= e^{-rt_0}(M(t_0)-m(t_0))+ e^{-r \tau(t_0)} m(\tau(t_0)) \tau'(t_0)\notag\\
&=  e^{-rt_0} (M(t_0)-m(t_0)) \left(1-e^{\beta(\tau(t_0)-t_0)}\right)<0,\quad \hbox{for a.e.}\ 0\le t_0<t_M. 
\end{align}
where the second line follows from \eqref{tau'} and the inequality is due to \eqref{tau>t_0}. This shows that $g(t_0)$, $0\le t_0\le t_M$, has a global minimum at $t_0=t_M$. 
Thus, it follows from \eqref{=g} that by taking $t_0=t_M$, $\alpha^*$ in \eqref{alpha^*}  attains $\inf_{\alpha\in\mathcal B}J(x,\alpha)=v(x)$, where the equality follows from Lemma \ref{pro:2.1}. Consequently, 
$
v(x) = J(x,\alpha^*)=g\left(t_M \right)=  \int_{0}^{t_M} e^{-rs}  M(s) ds, 
$
 where the last equality follows from \eqref{g} and \eqref{tau>t_0}. Finally, 
simply because $\tau(t_M)=t_M$ (again, by \eqref{tau>t_0}), one can without loss of generality take $\alpha^*_t=M(t)$ for all $0\le t\le T$.  
\end{proof}

Finally, consider the intermediate case of an initial balance $x>0$ small enough that maximum payments can pay off the debt by time $T$, but also large enough that minimum payments cannot pay it off by time $T$. 

\begin{proposition}\label{prop:x intermediate}
Let $\underline x< x\le \overline x$, $t_\oc\in [0,T)$ as in \eqref{eq:critical}, and define $x_\oc\in [\underline x,\overline x)$ as
\begin{equation}\label{wtx}
x_\oc:=  \int_{0}^{t_\oc} e^{-(r+\beta)s}  M(s) ds + \int_{t_\oc}^{T} e^{-(r+\beta)s}  m(s) ds
.
\end{equation}
\begin{enumerate}
\item
If $x> x_\oc$,  then 
\[
v(x) = v_1(x)\wedge v_2(x),
\]
where $v_1$ and $v_2$ are defined as in \eqref{v_1} and \eqref{v_2}, respectively. Furthermore, if $v_1(x)< v_2(x)$, $\alpha^*\in\mathcal A$ defined in \eqref{alpha^* x large} is an optimal control; otherwise, $\alpha^*\in\mathcal A$ defined by $\alpha^*_t=M(t)$, $0\le t\le T$, is an optimal control. 
\item
If $x\le x_\oc$, then 
$v(x) =v_2(x)$ with $v_2$ defined as in \eqref{v_2}.
Moreover, $\alpha^*\in\mathcal A$ defined by $\alpha^*_t=M(t)$, $0\le t\le T$, is an optimal control. 
\end{enumerate}
\end{proposition}

\begin{proof}
As $\underline x < x\le \overline x$ and $0<m(t)<M(t)$, there exists a unique $\tilde t\in (0, T]$ such that 
\begin{equation}\label{t^*}
x= \int_{0}^{\tilde t} e^{-(r+\beta)s}  M(s) ds + \int_{\tilde t}^{T} e^{-(r+\beta)s}  m(s) ds.
\end{equation}
Thus, $\tilde t\le t_M$ by \eqref{t^*} and the definition of $t_M>0$ in \eqref{v_2}.  
Now, decompose $\mathcal B$ in \eqref{B} into $\B_1:= \{\alpha\in\B:b^{\alpha}_T>0\}$ and $\B_2:= \{\alpha\in\B:b^{\alpha}_T\le 0\}$. 
In view of Remark~\ref{thm:1.1} and \eqref{t^*}, 
\begin{equation}\label{B_1, B_2}
\B_1= \{\alpha\in\B: 0\le t_0< \tilde t \}\quad \hbox{and}\quad \B_2= \{\alpha\in\B: \tilde t\le t_0\le T\}.
\end{equation}
For any $\alpha\in \B_1$, argue as in \eqref{=f(t^*)} to obtain that $J(x,\alpha) = f(t_0)$, where $f:\R\to\R$ is defined as in \eqref{f}. As shown after \eqref{f}, $f(t)$ is strictly decreasing for $t< T+ \frac{\ln \omega}{\beta}$ and strictly increasing for $t>T+ \frac{\ln \omega}{\beta}$. Thus, \eqref{B_1, B_2} implies that 
\begin{equation}\label{B_1 value}
\inf_{\alpha\in\B_1} J(x,\alpha)= f\left(t_\oc\wedge \tilde t\right),
\end{equation}
where $t_\oc\in[0,T)$ is defined as in \eqref{eq:critical}. 
For any $\alpha\in \B_2$, argue as in \eqref{=g} to obtain $J(x,\alpha) = g(t_0\wedge t_M)$, where $g:[0,t_M]\to\R$ is defined as in \eqref{g}. As shown below \eqref{g}, $g(t)$ is strictly decreasing for $t< t_M$. Thus, \eqref{B_1, B_2} and $\tilde t\le t_M$ imply that 
\begin{equation}\label{B_2 value}
\inf_{\alpha\in\B_2} J(x,\alpha) = g(t_M). 
\end{equation}
In view of $\tilde t\le t_M$, note also that 
$
g(t_M) \le g(\tilde t) = f(\tilde t),
$
{} where the equality follows from the definitions of $f$ and $g$ (\eqref{f} and \eqref{g}) and \eqref{t^*}. 
Now, by Lemma \ref{pro:2.1}, \eqref{B_1 value}, and \eqref{B_2 value}, 
\[
v(x) = \inf_{\alpha\in\B} J(x,\alpha) = f\left(t_\oc\wedge \tilde t\right)\wedge g(t_M)=
\begin{cases}
f\left(t_\oc\right)\wedge g(t_M)= v_1(x)\wedge v_2(x), &\quad \hbox{if}\ t_\oc<\tilde t,\\
g(t_M)=v_2(x), &\quad \hbox{if}\ t_\oc\ge \tilde t,
\end{cases} 
\]
where the third equality exploits $g(t_M) \le f(\tilde t)$, and $v_1$ and $v_2$ are defined as in \eqref{v_1} and \eqref{v_2}, respectively. Note from \eqref{t^*} and \eqref{wtx} that $t_\oc<\tilde t$ if and only if $x_c<x$. The desired result thus follows from the previous equality.
\end{proof}

In summary, Propositions~\ref{prop:x large}, \ref{prop:x small}, and \ref{prop:x intermediate} altogether demonstrate that the optimal strategy must be either $\alpha^{1}_t := M(t) 1_{[0,t_\oc]}(t)+m(t) 1_{(t_\oc,T]}(t)$ or $\alpha^{2}_t := M(t)$. Specifically, 
(i) if $\alpha^{1}$ can pay off the balance $x$ by time $T$ (i.e. $x\le x_\oc$), Proposition~\ref{prop:x small} and \ref{prop:x intermediate} (ii) state that it is best to pay off the debt as soon as possible, i.e. $\alpha^2$ is optimal;
(ii) if $\alpha^{1}$ cannot pay off the balance $x$ by time $T$ but $\alpha^2$ can (i.e. $x_\oc< x\le \overline x$), Proposition~\ref{prop:x intermediate} (i) states that one needs to compare the costs $v_1(x)$ and $v_2(x)$ to determine which one of $\alpha^1$ and $\alpha^2$ is optimal;
(iii) if $\alpha^{2}$ cannot pay off the balance $x$ by time $T$ (i.e. $x> \overline x$), Proposition~\ref{prop:x large} states that $\alpha^1$ is optimal. 
Thus, for the case $x_c< x\le \overline x$, we need to analyze $v_1(x)$ and $v_2(x)$ further to determine the optimal strategy. 

\begin{lemma}\label{lem:v1-v2}
(i)
$v_1(x)-v_2(x)$ is strictly decreasing on $[\hat x, \overline x]$, where
\begin{equation}\label{hat x}
\hat x:= \int_0^{t_\oc} e^{-(r+\beta)s} M(s) ds\in [0,\overline x). 
\end{equation}
(ii)
There exists a unique $x^*\in (\hat x,\overline x)$ such that $v_1(x^*)= v_2(x^*)$. Hence, $v_1(x)>v_2(x)$ for $x\in [\hat x,x^*)$ and $v_1(x)<v_2(x)$ for $x\in (x^*,\overline x]$. 
Moreover, $x^*$ is identified as in Theorem \ref{th:main} and satisfies $x^*>x_\oc$, with $x_\oc$ defined as in \eqref{wtx}. 
\end{lemma}

\begin{proof}
In view of \eqref{v_2}, $t_M>0$ is in fact a function of $x$, and is strictly increasing by definition. Henceforth, denote $t_M$ as $t_M(x)$ for clarity. As $t_M(\hat x) =t_\oc$ by construction, it follows that $t_M(x)>t_\oc$ for all $\hat x<x \le \overline x$. Thanks to the definitions of $v_1$ and $v_2$ in \eqref{v_1} and \eqref{v_2}, as well as the relation \eqref{t_2}, for any $\hat x\le x \le \overline x$, 
\begin{align}
v_1(x)-v_2(x) &= -\int_{t_\oc}^{t_M(x)} e^{-rs}  M(s) ds+ \int_{t_\oc}^{T} e^{-rs}  m(s) ds\nonumber\\
& + \omega e^{\beta T}\bigg(\int_{t_\oc}^{t_M(x)} e^{-(r+\beta)s}  M(s) ds  -\int_{t_\oc}^{T} e^{-(r+\beta)s}  m(s) ds \bigg)\nonumber &\\
&= -\int_{t_\oc}^{t_M(x)} e^{-rs}  M(s) \left(1-\omega e^{\beta(T-s)}\right) ds+ \int_{t_\oc}^{T} e^{-rs}  m(s) \left(1-\omega e^{\beta(T-s)}\right) ds. \label{v1-v2}
\end{align}
Note that $s> t_\oc= (T+\frac{\ln w}{\beta})^+$ if and only if 
$
1-\omega e^{\beta(T-s)}>0.
$
 Hence, because $M(s)$ and $1-\omega e^{\beta(T-s)}$ are strictly positive and $t_M(x)$ is strictly increasing, \eqref{v1-v2} implies that $v_1(x)-v_2(x)$ is strictly decreasing on $[\hat x, \overline x]$. Now, observe from \eqref{t_2} that $t_M(\overline x)= T$ and from \eqref{v1-v2} that 
\[
v_1(\overline x)-v_2(\overline x) = -\int_{t_\oc}^{T} e^{-rs}  (M(s)-m(s)) \left(1-\omega e^{\beta(T-s)}\right) ds <0,
\]
because $M(s)>m(s)$ and $1-\omega e^{\beta(T-s)}>0$. On the other hand, by $t_M(\hat x) =t_\oc$, \eqref{v1-v2} yields $v_1(\hat x)-v_2(\hat  x) = \int_{t_\oc}^{T} e^{-rs}  m(s) \left(1-\omega e^{\beta(T-s)}\right) ds >0$, again because $1-\omega e^{\beta(T-s)}>0$. As $v_1(x)-v_2(x)$ is strictly decreasing on $[\hat x, \overline x]$, there must exist $x^*\in (\hat x, \overline x)$ such that  $v_1(x^*)-v_2(x^*)=0$. Note that $x^*$ is identified by setting the right hand side of \eqref{v1-v2} to be zero, which leads to the characterization in Theorem \ref{th:main}. Now, in view of \eqref{hat x}, \eqref{wtx}, and \eqref{x min, max}, $\hat x < x_\oc<\overline x$ by definition. Observe from \eqref{v_1}, \eqref{wtx}, and \eqref{g} that
\[
v_1(x_\oc) = \int_{0}^{t_\oc} e^{-rs}  M(s) ds+ \int_{t_\oc}^{T} e^{-rs}  m(s) ds = g(t_\oc) > g(t_M(x_\oc))= v_2(x_\oc),
\]
where the inequality follows from the fact that $g:[0,t_M(x_\oc)]\to\R$ is minimized at $t_M(x_\oc)$ (recall \eqref{g'}) and the last equality is due to \eqref{g} and \eqref{v_2}. This fact readily implies $x_\oc<x^*$. 
\end{proof}

\begin{proof}[Proof of Theorem \ref{th:main}]
By Lemma~\ref{lem:v1-v2}, $x^*>x_\oc$, $v_1(x)>v_2(x)$ for $x\in [\hat x,x^*)$, and $v_1(x)<v_2(x)$ for $x\in (x^*,\overline x]$. The results of Proposition \ref{prop:x intermediate} are then simplified to $v(x) = v_1(x)$ for $x\in (x^*,  \overline x]$ and $v(x) = v_2(x)$ for $x\in (\overline x, x^*].$ 
Combining this fact with Propositions~\ref{prop:x large} and \ref{prop:x small} yields the claim. 
\end{proof}


\section{Proofs for Section~\ref{sec:simple}}\label{sec:proofsnew}

\def \fb{{\mathfrak b}}
\def \fJ{{\mathfrak J}}

To prove Lemma~\ref{lem:reduce to B}, we need to introduce additional notation. For any fixed $\ell\in [0,T)$, consider
\begin{equation}\label{F^elln}
\mathcal A^\ell := \{ \alpha:\text{$t\mapsto \alpha_t$ Lebesgue measurable, } m(t+\ell) \leq \alpha_t \leq M(t+\ell)\ \hbox{for a.e.}\ 0\le t\le T-\ell\}.
\end{equation}
Clearly, $\A^0 =\A$. Given $x>0$ and $\alpha\in\A^\ell$, we consider $\tau^\ell$ as in \eqref{tau}, with $T>0$ therein replaced by $T-\ell>0$, and $J^\ell(x,\alpha)$ as in \eqref{J}, with $\tau$ therein replaced by $\tau^\ell$. 

\begin{remark}\label{rem:J and J^ell}
Fix $x>0$ and $\alpha\in\A$. For any $0\le \ell<\tau$, set $y:= b^\alpha_{\ell}$ and define $\alpha^\ell\in\A^\ell$ by 
\begin{equation}\label{alpha^ell}
\alpha^\ell_s:=\alpha_{\ell+s}\quad \hbox{for $s\in[0,T-\ell]$}. 
\end{equation}
Then, in view of \eqref{J}, 
\begin{align*}
J(x,\alpha) &= \int_0^\ell e^{-rs} \alpha_s ds + \int_\ell^{\tau} e^{-rs} \alpha_s ds + e^{-r\tau} \omega b^\alpha_{\tau}\\
&= \int_0^\ell e^{-rs} \alpha_s ds + e^{-r\ell} \bigg(\int_0^{\tau^\ell} e^{-rs} \alpha^\ell_s ds + e^{-r\tau^\ell} \omega b^{\alpha^\ell}_{\tau^\ell}\bigg) = \int_0^\ell e^{-rs} \alpha_s ds +  e^{-r\ell} J^\ell(y,\alpha^\ell).
\end{align*}
\end{remark}

\begin{proof}[Proof of Lemma~\ref{lem:reduce to B}]
(i) Suppose that $\theta(\alpha) = T$. Observe from \eqref{t_alpha} that $\theta(m) \ge \theta(\alpha)=T$. That is, neither $\alpha_s$ nor $m(s)$ pays off the accumulated interest balance by time $T$, so that the principal remains untouched at time $T$. It follows that
\begin{align}
J(x,\alpha) -J(x,m) &= \int_0^T e^{-rt} \alpha_s ds + e^{-rT}\omega \left(x+ (r+\beta)xT - \int_0^T \alpha_s ds\right) \notag\\
&\hspace{0.2in}- \int_0^T e^{-rt} m(s) ds -e^{-rT}\omega \left(x+ (r+\beta)xT - \int_0^T m(s) ds\right) \notag\\
& = \int_0^T \left(e^{-rt} - \omega e^{-rT}\right) (\alpha_s - m(s)) ds \ge 0,\label{J-J}
\end{align} 
where we used the fact $\omega\in (0,1)$ for the inequality. 

(ii) Suppose that $\theta(\alpha) < T$. If $\alpha_s=m(s)$ for a.e.\ $s\in[0,\theta(\alpha)]$, $t_0 = \theta(\alpha)$ in \eqref{payment equal}. If $\alpha_s=M(s)$ for a.e.\ $s\in[0,\theta(\alpha)]$, $t_0=0$ in \eqref{payment equal}. If both $\{s\in[0,\theta(\alpha)]:\alpha_s> m(s)\}$ and $\{s\in[0,\theta(\alpha)]:\alpha_s< M(s)\}$ have positive measures, consider $f:[0,\theta(\alpha)] \to \R$ defined by
\begin{align*}
f(t) := \int_{0}^{t}  (\alpha_s-m(s))  ds - \int_{t}^{\theta(\alpha)} (M(s)-\alpha_s)  ds. 
\end{align*}
As $m(s)\le \alpha_s\le M(s)$ and $m(s)<M(s)$ for all $s\in[0,\theta(\alpha)]$, $f$ is by definition continuous and strictly increasing. Since $\{s\in[0,\theta(\alpha)]:\alpha_s< M(s)\}$ has a positive measure, $f(0) =  - \int_{0}^{\theta(\alpha)}(M(s)-\alpha_s) ds<0$. Likewise, the fact that $\{s\in[0,\theta(\alpha)]:\alpha_s> m(s)\}$ has a positive measure implies $f(\theta(\alpha)) = \int_{0}^{\theta(\alpha)} (\alpha_s-m(s)) ds>0$. Hence, there exists a unique $0<t_0<\theta(\alpha)$ such that $f(t_0)=0$, which is equivalent to \eqref{payment equal}.  
Observe that
\begin{align*}
\int_{0}^{t_0} e^{-rs} (\alpha_s-m(s))  ds \ge \int_{0}^{t_0}  e^{-rt_0}(\alpha_s-m(s))  ds &= \int_{t_0}^{\theta(\alpha)} e^{-rt_0} (M(s)-\alpha_s)  ds\\
& \ge  \int_{t_0}^{\theta(\alpha)} e^{-rs} (M(s)-\alpha_s)  ds,
\end{align*}
where the equality follows from \eqref{payment equal} (with both sides multiplied by $e^{-rt_0}$). This gives
\begin{equation}\label{cost smaller}
\int_{0}^{\theta(\alpha)} e^{-rs} \alpha_s ds \ge \int_{0}^{\theta(\alpha)} e^{-rs} \left(m(s)1_{[0, t_0]}(s) + M(s) 1_{(t_0,\theta(\alpha)]}(s)\right)  ds.
\end{equation}

Now, we claim that $\theta(\overline\alpha) = \theta(\alpha)$. By \eqref{payment equal} and the definition of $\overline\alpha\in\A$ in \eqref{mMalpha}, we have
\begin{equation}\label{payment equal'}
\int_0^{\theta(\alpha)}\alpha_s ds=\int_0^{\theta(\alpha)}\overline\alpha_s ds.
\end{equation}
Suppose that $\theta(\overline\alpha)> \theta(\alpha)$. By \eqref{payment equal'}, $p^{\overline{\alpha}}_{\theta(\alpha)} = p^\alpha_{\theta(\alpha)}  = x$. This, together with $\overline\alpha_s = \alpha_s$ for all $s\in(\theta(\alpha),T]$ (see \eqref{mMalpha}), implies that $p^{\overline{\alpha}}_{s} = p^\alpha_{s}$ for all $s\in[\theta(\alpha),T]$. By the definition of $\theta(\alpha)$ in \eqref{t_alpha}, there exist $\{t_n\}_{n\in\N}$ in $(\theta(\alpha),T]$ with $t_n\downarrow \theta(\alpha)$ such that $p^{\overline\alpha}_{t_n} = p^{{\alpha}}_{t_n} < x$ for all $n\in\N$. This shows that $\theta(\overline\alpha)\le \theta(\alpha)$, a contradiction. On the other hand, suppose that $\theta(\overline\alpha)< \theta(\alpha)$, which implies 
\begin{equation}\label{payment equal''}
\int_0^{\theta(\overline\alpha)}\overline\alpha_s ds = (r+\beta)x \theta(\overline\alpha) \ge \int_0^{\theta(\overline\alpha)}\alpha_s ds.
\end{equation}
If $\theta(\overline\alpha)< t_0$, since $\overline\alpha_s = m(s) \le \alpha_s$ for all $s\in [0,t_0]$, we must have $\theta(\alpha)\le \theta(\overline\alpha)$, a contradiction. With $\theta(\overline\alpha)\ge t_0$, \eqref{payment equal''} and \eqref{payment equal'} together yield $\int_{\theta(\overline\alpha)}^{\theta(\alpha)}\alpha_s ds \ge \int_{\theta(\overline\alpha)}^{\theta(\alpha)}\overline\alpha_s ds = \int_{\theta(\overline\alpha)}^{\theta(\alpha)}M(s)ds$. As $\alpha\in\A$, we must have $\alpha_s = M(s)$ for a.e.\ $s\in[\theta(\overline\alpha),\theta(\alpha)]$, so that $\int_{\theta(\overline\alpha)}^{\theta(\alpha)}\alpha_s ds = \int_{\theta(\overline\alpha)}^{\theta(\alpha)}\overline\alpha_s ds$. In view of \eqref{payment equal'}, this shows that \eqref{payment equal''} in fact holds as an equality. It follows that $p^\alpha_{\theta(\overline\alpha)} = p^{\overline{\alpha}}_{\theta(\overline\alpha)} = x$. This, along with $\alpha_s = M(s)=\overline\alpha_s$ for a.e.\ $s\in[\theta(\overline\alpha),\theta(\alpha)]$, implies that $p^\alpha_{s} = p^{\overline{\alpha}}_{s}$ for all $s\in[\theta(\overline\alpha),\theta(\alpha)]$. By the definition of $\theta(\overline\alpha)$ in \eqref{t_alpha}, there exist $\{t_n\}_{n\in\N}$ in $(\theta(\overline\alpha),\theta(\alpha)]$ with $t_n\downarrow \theta(\overline\alpha)$ such that $p^\alpha_{t_n} = p^{\overline{\alpha}}_{t_n}<x$ for all $n\in\N$. This shows that $\theta(\alpha)\le \theta(\overline\alpha)$, a contradiction. Hence, we conclude that $\theta({\overline\alpha}) = \theta(\alpha)$. 

Finally, recall the notation introduced above Remark~\ref{rem:J and J^ell}. Since \eqref{t_alpha} and \eqref{b} imply $b^\alpha_{\theta(\alpha)} = b^\alpha_0=x$, we deduce from Remark~\ref{rem:J and J^ell} that 
\begin{align*}
J(x,\alpha) &= \int_0^{\theta(\alpha)} e^{-rs}\alpha_s ds + e^{-r \theta(\alpha)} J^{\theta(\alpha)}(x,\alpha^{\theta(\alpha)})\\
& \ge \int_0^{\theta(\alpha)} e^{-rs}\overline\alpha_s ds + e^{-r \theta(\alpha)} J^{\theta(\alpha)}(x,\overline\alpha^{\theta(\alpha)})=J(x,\overline\alpha), 
\end{align*}
where the inequality follows from \eqref{cost smaller} and the definition of $\overline\alpha$ in \eqref{mMalpha}, and the last equality is due to $\theta({\overline\alpha})= \theta(\alpha)$ and Remark~\ref{rem:J and J^ell}. 
\end{proof}

\begin{proof}[Proof of Corollary \ref{cor:verylarge}]
(i) For any $\alpha\in\A$, since $\theta(\alpha)\ge \theta(M)=T$ implies $\theta(\alpha)=T$, we conclude from Lemma~\ref{lem:reduce to B} (i) that $J(x,m)\le J(x,\alpha)$. It follows that $J(x,m) = \inf_{\alpha\in\A} J(x,\alpha) = v(x)$, as desired. 

(ii) As $\theta(m) = 0$, there exists $\{t_n\}_{n\in\N}$ in $(0,T]$ with $t_n\downarrow 0$ such that $m(t_n) > (r+\beta) x$ for all $n\in\N$. For any $\alpha\in\A$, $\alpha_t \ge m(t) \ge m(t_n) > (r+\beta) x$ for $t\ge t_n$ because $m(t)$ is nondecreasing.
As $t_n\downarrow 0$, it follows that $\alpha_t > (r+\beta) x$ for all $t\in (0,T]$. Now, consider the process $\fb^\alpha$ defined by $d \fb^\alpha_t = ((r+\beta)\fb^\alpha_t-\alpha_t) dt$, $t\in(0,T]$, with $\fb^\alpha_0 =x$, as well as $\mathfrak{p}^\alpha_t := \inf_{0\le s\le t}\fb^\alpha_s$, $t\in[0,T]$. Thanks to $\alpha_t > (r+\beta) x$ for all $t\in (0,T]$, $t\mapsto\fb^\alpha_t$ is strictly decreasing, which in turn implies $\mathfrak{p}^\alpha_t = \fb^\alpha_t$ for all $t\in[0,T]$. It follows that $(\fb^\alpha,\mathfrak{p}^\alpha)$ satisfies \eqref{b}-\eqref{p} and is in fact the unique solution (as the coefficients in \eqref{b}-\eqref{p} are Lipschitz). Moreover, the dynamics reduces to \eqref{eq:1} due to $\mathfrak{p}^\alpha= \fb^\alpha$.  As such a reduction to \eqref{eq:1} holds for all $\alpha\in\A$, an optimal strategy $\alpha^*\in\A$ follows from Theorem~\ref{th:main}. 
\end{proof}

Before we proceed to prove Lemma \ref{lem:p decreases}, notice that it is essentially an extension of Lemma~\ref{pro:2.1} to the present case of simple interest. The goal is to establish an inequality similar to \eqref{toshow}, but the challenge here is that, in contrast to the case of compound interest, the total balance process $b$ is no longer tractable. Yet, a careful comparison between the total balances in these two cases shows that the desired inequality still holds with simple interest. 

In the following, to clearly distinguish between the two total balance processes (one in \eqref{eq:1} with compound interest, the other in \eqref{b}-\eqref{p} with simple interest), we will reserve ``$b$'' for the latter and use ``$\fb$'' to denote the former. 

\begin{remark}\label{rem:b<fb}
By the comparison theorem of ordinary differential equations, we deduce from the dynamics \eqref{b}-\eqref{p} (satisfied by $b$) and \eqref{eq:1} (satisfied by $\fb$) that $b^\alpha_t \le \fb^\alpha_t$, $t\ge 0$, for any $\alpha\in\A$. 
\end{remark}

\begin{proof}[Proof of Lemma \ref{lem:p decreases}]
Thanks to $\alpha\in \A\setminus \B_{[a,c]}$, the same arguments in the proof of Lemma~\ref{pro:2.1} (see \eqref{t_0} particularly) show that there exists $s_0\in (a,c)$ such that
\begin{align}\label{s_0 exists}
\int_{a}^{s_0} e^{-(r+\beta) t} M(t) dt + \int_{s_0}^c e^{-(r+\beta) t} m(t) dt= \int_a^c e^{-(r+\beta) t} \alpha_t dt.
\end{align}
In addition, by an estimation similar to \eqref{<}, this implies
\begin{equation}\label{heha}
\int_a^{s_0} e^{-r t} M(t) dt + \int_{s_0}^c e^{-r t} m(t) dt < \int_a^c e^{-r t} \alpha_t dt.
\end{equation}

Now, consider $\alpha_{(s_0)}\in\A$ defined as in \eqref{alMmal} (with $u=s_0$). 
Set $y := b^{\alpha}_a>0$. In the following, we will add the superscript ``$y$'' to the functions $t\mapsto b_t$ (resp.\ $t\mapsto \fb_t$) to emphasize its initial condition $b_0=y$ (resp.\ $\fb_0=y$). Observe that
\[
b^{y,(\alpha_{(s_0)})^a}_{c-a} \le \fb^{y,(\alpha_{(s_0)})^a}_{c-a} = \fb^{y,\alpha^a}_{c-a} = b^{y,\alpha^a}_{c-a}, 
\]
where we use the notation \eqref{alpha^ell} (with $\ell=a$). Note that the inequality above follows from Remark~\ref{rem:b<fb}, and the first equality above is due to \eqref{s_0 exists} and the explicit formula \eqref{b formula} (now satisfied by $\fb$). The second equality above stems from the fact that $t\mapsto p^\alpha_t$ is strictly decreasing on $[a,c]$, so that the dynamics in \eqref{b}-\eqref{p} and \eqref{eq:1} (now satisfied by $\fb$) coincide on $[a,c]$. 

Let us finish the proof by dealing with the following two cases separately. For the case where $b^{y,(\alpha_{(s_0)})^a}_{c-a} = b^{y,\alpha^a}_{c-a}$, we have
\begin{align}\label{baralpha>}
J(x,\alpha) &= \int_0^\tau e^{-r t} \alpha_t dt + e^{-r\tau}\omega b^{x,\alpha}_{\tau}\notag\\
&> \int_0^{a} e^{-r t} \alpha_t dt + \int_{a}^{s_0} e^{-r t} M(t) dt + \int_{s_0}^c e^{-r t} m(t) dt+ \int_{c}^{\tau} e^{-r t} \alpha_t dt + e^{-r\tau}\omega b^{y,(\alpha_{(s_0)})^a}_{\tau-a}\notag\\
&= J(x,\alpha_{(s_0)}). 
\end{align}
where the inequality follows from \eqref{heha} and $b^{y,(\alpha_{(s_0)})^a}_{t-a} = b^{y,\alpha^a}_{t-a}= b^{x,\alpha}_t$ for $t\ge c$, and the last equality holds by the definition of $\alpha_{(s_0)}$ in \eqref{alMmal}. Hence, we obtain $J(x,\alpha_{(u)}) < J(x,\alpha)$ by taking $u=s_0$. 
Now, consider the other case where $b^{y,(\alpha_{(s_0)})^a}_{c-a} < b^{y,\alpha^a}_{c-a}$. Observe from \eqref{alMmal} that $(\alpha_{(a)})_t = m(t)$ for all $t\in(a, c]$. It follows that
\[
b^{y,(\alpha_{(a)})^a}_{c-a} = b^{y,m^a}_{c-a} > b^{y,\alpha^a}_{c-a} > b^{y,(\alpha_{(s_0)})^a}_{c-a},
\]
where the first inequality follows from $\alpha_t \ge m(t)$ for all $t\in[a,c]$ and that $\{t\in[a,c] : \alpha_t > m(t)\}$ has a positive measure (as $\alpha\notin \B_{[a,c]}$). Then, by the continuity of $u\mapsto b^{y,(\alpha_{(u)})^a}_{c-a}$, there exists $u\in (a, s_0)$ such that $b^{y,(\alpha_{(u)})^a}_{c-a} = b^{y,\alpha^a}_{c-a}$. 
Note that \eqref{heha} still holds with $s_0$ therein replaced by $u$. Indeed, with $u\in (a, s_0)$, the left-hand side becomes even smaller with $u$ in place of $s_0$. It follows that the same calculation in \eqref{baralpha>} (with $s_0$ replaced by $u\in (a, s_0)$) still holds, which leads to $J(x,\alpha)>J(x,\alpha_{(u)})$. 
\end{proof}

\section{Conclusion}\label{sec:conc}

Federal student loans are complex debt contracts that enable borrowers to cap repayments to a fraction of their income, have the balance forgiven after at least twenty years in good standing, and accrue interest without capitalization while the loan is in negative amortization. On one hand, the prospects of simple accrued interest and debt forgiveness entices borrowers to enroll in income-driven schemes to minimize payments and maximize the forgiven amount. On the other hand, the high interest rate on student loans, combined with the long forgiveness horizon, may more than offset the benefits of income-driven repayments for small balances.

This paper finds the cost-minimizing repayment strategy that accounts for income-driven schemes and debt forgiveness. For very large or very small loan balances, our treatment also incorporates the effect of simple accrued interest with negative amortization.
A small loan should be paid as soon as possible, i.e., maximizing repayments until it is paid off. For a very large loan, enrollment in income-driven schemes is optimal, as the benefits of simple accrued interest and balance forgiveness override interest costs. 

\section*{Acknowledgments}
We thank Gur Huberman, Sara Biagini, and two anonymous referees of SIFIN for stimulating comments. We are especially grateful to Erik Kroll of Hilltop Financial Advisors for useful discussions and to the  editor-in-chief of SIFIN, Mete Soner, for encouraging us to deepen our analysis of student loans.

\bibliographystyle{siamplain}
\bibliography{studentloans}

\end{document}